\documentclass[a4paper,onecolumn,unpublished,11pt]{quantumarticle}
\pdfoutput=1
\usepackage[english]{babel}
\usepackage[numbers,sort&compress]{natbib}
\usepackage{amsmath}
\usepackage{amssymb}
\usepackage{amsthm}
\usepackage{graphicx}
\usepackage{mathrsfs}
\usepackage{enumitem}
\usepackage{csquotes}
\usepackage{braket}
\usepackage{mathtools}
\usepackage{booktabs}
\usepackage{thmtools}
\usepackage{algorithm}

\usepackage{algpseudocode}
\usepackage{verbatim}
\usepackage{soul}
\usepackage{apptools}
\usepackage[most]{tcolorbox}
\usepackage{tikz}
\usepackage[svgnames]{xcolor}
\usetikzlibrary{arrows.meta}

\providecommand{\myvec}[1]{\ensuremath{\boldsymbol{#1}}}

\providecommand{\xx}{\ensuremath{\myvec{x}}}

\providecommand{\calO}{\ensuremath{\mathcal{O}}}

\newtheorem{theorem}{Theorem}

\newtheorem{proposition}[theorem]{Proposition}
\newtheorem{definition}[theorem]{Definition}
\newtheorem{corollary}[theorem]{Corollary}
\newtheorem{observation}[theorem]{Observation}
\newtheorem{remark}[theorem]{Remark}

\newcommand{\Abs}[1]{\left\lvert #1 \right\rvert}

\newcommand{\maxlin}{\textup{max-LINSAT}} 
\newcommand{\maxxor}{\textup{max-XORSAT}} 
\newcommand{\maxkxor}{\textup{max-$k$-XORSAT}} 
\newcommand{\maxeklinq}[2]{\textup{max-E}#1\textup{-LIN-}#2}

\newcommand{\maxeksat}[1]{\textup{max-E}#1\textup{-SAT}}

\newcommand{\maxekxor}[1]{\textup{max-E}#1\textup{-XORSAT}}
\newcommand{\maxlinqr}[2]{\textup{max-LINSAT(}#1\textup{,}#2\textup{)}}
\newcommand{\maxeklinsat}[3]{\textup{max-E}#1\textup{-LINSAT}(#2,#3)}

\newcommand{\NP}{\mathsf{NP}}
\newcommand{\Pclass}{\mathsf{P}}
\newcommand{\RP}{\mathsf{RP}}
\newcommand{\coRP}{\mathsf{coRP}}
\newcommand{\BQP}{\mathsf{BQP}}
\newcommand{\BPP}{\mathsf{BPP}}

\newcommand{\OPT}{\mathrm{OPT}}

\newif\ifverbose
\verbosetrue

\definecolor{max}{RGB}{164,12,52}

\usepackage[most]{tcolorbox}

\tcbset {
  base/.style={
    arc=0mm, 
    bottomtitle=0.5mm,
    boxrule=0mm,
    colbacktitle=black!10!white, 
    coltitle=black, 
    fonttitle=\bfseries, 
    left=2.5mm,
    leftrule=1mm,
    right=3.5mm,
    title={#1},
    toptitle=0.75mm, 
  }
}

\definecolor{brandblue}{rgb}{0.34, 0.7, 1}
\newtcolorbox{mainbox}[1]{
  colframe=brandblue, 
  base={#1}
}

\newtcolorbox{subbox}[1]{
  colframe=black!30!white,
  base={#1}
}

\newcommand{\dccqs}{Dahlem Center for Complex Quantum Systems, Freie Universit{\"a}t Berlin, 14195 Berlin, Germany}
\newcommand{\hzb}{Helmholtz-Zentrum Berlin f{\"u}r Materialien und Energie, 14109 Berlin, Germany}
\newcommand{\hhi}{Fraunhofer Heinrich Hertz Institute, 10587 Berlin, Germany}
\newcommand{\tub}{Electrical Engineering and Computer Science, Technische Universität Berlin, 10587 Berlin, Germany}

\usepackage[bookmarks=true,colorlinks=true,linkcolor=teal,citecolor=teal,urlcolor=teal,bookmarksopen=true,bookmarksopenlevel=1]{hyperref}
\usepackage[capitalize]{cleveref}

\begin{document}

\title{Approximability limits for bounded-degree \newline max-LINSAT and implications for decoded quantum interferometry}

    \author{Maximilian\ J.\ Kramer}
    \email{m.kramer@fu-berlin.de}
    \affiliation{\dccqs}

    \author{Carsten\ Schubert}
    \affiliation{\tub}

    \author{Jens\ Eisert}
    \affiliation{\dccqs}
    \affiliation{\hzb}
    \affiliation{\hhi}

\begin{abstract}
For general \maxkxor{} with $k \geq 3$, no polynomial-time algorithm can do substantially better than random guessing on worst-case instances unless $\mathsf{P} = \mathsf{NP}$: approximating beyond the random-assignment value of $1/2$ is $\mathsf{NP}$-hard. The picture changes when each variable appears in at most $D$ constraints. In that bounded-degree setting, polynomial-time algorithms can provably beat the random baseline by an additive amount of order $1/\sqrt{D}$. For Boolean instances, this scaling is known to be optimal: the matching hardness result is due to Trevisan, while the corresponding algorithmic guarantee was established by Barak et al.
Whether the same holds over general finite fields, and what it implies for quantum algorithms, has not been established.
We make this connection explicit and extend the hardness to \maxeklinsat{$k$}{$q$}{$r$} with bounded degree $D$ and over arbitrary finite fields $\mathbb{F}_q$, proving that it is $\NP$-hard to exceed $r/q + \calO_{q,r}(1/\sqrt{D})$. These results provide the complexity-theoretic benchmark for the bounded-degree instances targeted by decoded quantum interferometry (DQI), QAOA, and classical heuristics.
Any quantum advantage on bounded-degree instances is therefore confined to the constant prefactor. We further show that in the context of DQI and on $(k,D)$-regular instances, this prefactor is sensitive to the nature of the decoder: DQI with classical decoders faces an information-theoretic $1/\sqrt{D \log D}$ barrier that prevents it from matching the hardness scaling, while DQI with quantum decoders is compatible with the $1/\sqrt{D}$ scaling---identifying quantum decoding as the key ingredient for matching the complexity-theoretic scaling with DQI.
\end{abstract}
\maketitle

\section{Introduction}

A central goal in the study of $\NP$-hard optimization problems is to characterize the best approximation ratios achievable by polynomial-time algorithms. Among such problems, sparse constraint satisfaction problems, and in particular \maxkxor{} and related systems of linear constraints over finite fields, play a particularly important role.
In \maxkxor{}, one is given a system of parity constraints involving at most $k$ Boolean variables each, and the goal is to find an assignment satisfying as many constraints as possible. This is a problem of substantial interest in its own right, but has also recently risen to prominence in the context of quantum algorithms for optimization, for which it is used as a testbed for achieving potential quantum advantages \cite{Jordan2024DQI}.

For general \maxkxor{} with $k \geq 3$, the random-assignment threshold $1/2$ is a tight inapproximability wall: H\r{a}stad~\cite{hastad2001} has proven that no polynomial-time algorithm can exceed this threshold on worst-case instances, assuming $\Pclass \neq \NP$. We have recently extended this inapproximability result to \maxlinqr{$q$}{$r$} over arbitrary finite fields $\mathbb{F}_q$ and acceptance sets of size $r$ with threshold $r/q$~\cite{kramer2026tight}. The instances produced by these PCP-reductions, however, have unbounded variable degree $D$ and variables to clauses ratio $n/m = o(1)$.

To put this into perspective, it is important to state that PCP-based inapproximability~\cite{AS98,ALMSS98} does not tell the full story: the picture changes qualitatively when instances possess \emph{exploitable structure}. In this case, both the approximability landscape and the question of quantum advantage change qualitatively. 
At one extreme, highly structured problems admit provable quantum--classical separations under standard complexity assumptions~\cite{Pirnay2024,szegedy2022,Yamakawa2024,buhrman2025formalframeworkquantumadvantage}.
At the other extreme, worst-case PCP instances have no exploitable structure at all, and quantum algorithms face the same inapproximability walls as classical ones~\cite{kramer2026tight}. 
The scientifically interesting regime lies between these extremes: instances with enough structure for algorithms to exploit, 
but where that structure arises naturally from the problem rather than being engineered to enable separation results.

\emph{Decoded quantum interferometry} (DQI)~\cite{Jordan2024DQI} and follow-up results \cite{
chailloux2024softdecoders, Patamawisut2025, ralli2025DQI, bu2025DQInoise, briaud2025quantumadvantagesolvingmultivariate, sabater2025solvingindustrialintegerlinear, marwaha2025complexitydecodedquantuminterferometry, anschuetz2025DQI, piveteau2025_quantum_decoding, parekh2025DQI_maxcut, gu2025algebraicgeometrycodesdecoded, kothari2025exponentialquantumspeedupmathrmsisinfty, schmidhuber2025hamiltoniandecodedquantuminterferometry, khattar2025verifiablequantumadvantageoptimized, chailloux2025opixsoftdecoders, rosmanis2026nearlylineartimedecodedquantum, bu2026hamiltoniandecodedquantuminterferometry,kramer2026tight, sun2026,shutty2026optimizationusinglocallyquantumdecoders,garrido2026,thelen2026constraintcodedqikit,krajenbrink2026} operate squarely in this regime as a novel and compelling quantum approach for approximating \maxxor{} and \maxlin{} more generally. 
DQI reduces an optimization problem to a (quantum) decoding problem via a Fourier transform over the underlying finite field, an idea originating in Regev's reduction~\cite{ATS03,Regev2004,AR2005,Regev2009}.
Its performance is governed by the so-called semicircle law which is a function of the decodable structure present in the instance~\cite{Jordan2024DQI}. This makes DQI a compelling framework for exploring quantum advantage in approximation: it cannot help on structureless worst-case instances~\cite{kramer2026tight}, but may outperform classical algorithms when sufficient decodable structure is present~\cite{Jordan2024DQI,gu2025algebraicgeometrycodesdecoded}. This is a promising avenue toward quantum advantage for \emph{practically relevant problems}~\cite{MindTheGaps,VastWorld,Abbas_2024}.

Bounded-degree \maxkxor{}-instances---where each variable appears in at most $D$ constraints---occupy precisely this intermediate regime. 
On the worst-case side, polynomial-time algorithms are known to exceed the random-assignment threshold by 
$\Omega(1/D)$~\cite{Hastad2000}, 
$\Omega(D^{-3/4})$~\cite{FGG14bounded}, and 
$\Omega(1/\sqrt{D})$~\cite{Barak2015} on arbitrary bounded-degree instances.
On the average-case side, Shutty et al.~\cite{shutty2026optimizationusinglocallyquantumdecoders} recently undertook a systematic study on instances drawn from Gallager's $(k,D)$-regular ensemble, showing that DQI/Regev with various decoders, QAOA~\cite{FGG14}, and classical heuristics such as Turbo Prange all achieve approximation ratios of the form $1/2 + c_k/\sqrt{D}$ with algorithm- and $k$-dependent constants $c_k$. 
Notably, even a relatively simple locally-quantum decoder (FGUM)  combined with Regev's reduction suffices to match the $1/\sqrt{D}$ scaling of the best classical heuristic~\cite{shutty2026optimizationusinglocallyquantumdecoders}, suggesting that quantum decoding is a promising route toward quantum advantage in this regime.

In the light of all this, a natural question is whether this shared $1/\sqrt{D}$ scaling is coincidental or fundamental, and whether any algorithm---quantum or classical---can do better. This is a crucial question: it provides the goalpost for bounded-degree algorithm design and determines the room for improvement by quantum or classical means. Understanding whether the $1/\sqrt{D}$ scaling is a hard ceiling or an artifact of current techniques is, therefore, important both for classical approximation theory and for assessing the potential of quantum approaches to combinatorial optimization.

In this work, we make three concrete technical contributions and interpret the implications. First, we make explicit a complexity-theoretic benchmark for worst-case bounded-degree \maxekxor{$k$} that follows from connecting two known results not previously brought into contact with the DQI literature: Trevisan's randomized degree-reduction technique~\cite{Trevisan2001} and the algorithmic result of Barak et al.~\cite{Barak2015}.
Together, these establish that the optimal worst-case approximation ratio is $1/2 + \Theta_k(1/\sqrt{D})$, which is tight for odd $k$ up to constant factors between known algorithms and hardness. This, in turn, confines any potential quantum advantage on bounded-degree \maxekxor{$k$} to the constant prefactor.
Second, we extend this bounded-degree hardness from \maxekxor{$k$} to \maxeklinsat{$k$}{$q$}{$r$} over arbitrary finite fields $\mathbb{F}_q$ with acceptance sets of size $r$, proving that the threshold $r/q + \calO_{q,r}(1/\sqrt{D})$ is $\NP$-hard to exceed.
Third, we derive information-theoretic ceilings on DQI's performance by extending the analysis of Ref.~\cite[Sections~13.1--13.2]{Jordan2024DQI} to $(k,D)$-regular \maxeklinsat{$k$}{$q$}{$r$}, showing that quantum decoders are information-theoretically necessary for DQI to match the complexity-theoretic scaling: DQI with classical decoders is limited to $r/q + \calO_{q,r}(1/\sqrt{D \log D})$, while DQI with quantum decoders is compatible with the $1/\sqrt{D}$ scaling.

We note that confinement of possible quantum advantage to the constant prefactor is not a limitation in practical settings: for all $(q,r)$, the approximation ratio lies in $[r/q, 1]$, so the constant directly determines the algorithmic advantage over random assignment. In this setting, a constant-factor improvement can be operationally significant.

\section{Preliminaries}
\subsection{Problem definitions}
We start by defining the optimization problems considered in this work, in particular focusing on the (bounded-degree) \maxlin-problem as defined in Ref.~\cite{Jordan2024DQI}. Throughout, we denote by $\OPT$ the maximum number of satisfiable constraints and define an approximation ratio $\alpha$ for a solution satisfying $s$ out of $m$ constraints to be $\alpha \coloneq s / m$.

\begin{definition}[\maxlin{} over $\mathbb{F}_q$]
    Let $\mathbb{F}_q$ be the finite field of order $q$ and let $B \in \mathbb{F}_q^{m \times n}$. For each $i=1,\ldots,m$, let $F_i \subset \mathbb{F}_q$ be an arbitrary subset of $\mathbb{F}_q$, which yields a corresponding constraint $\sum_{j=1}^n B_{i, j} x_j \in F_i$. The \maxlin{} task is to find $\xx \in \mathbb{F}_q^n$ satisfying as many as possible of these $m$ constraints. The \emph{arity} of constraint $i$ 
    is the number of nonzero entries in row $i$ of $B$. The \emph{degree} of variable $x_j$ is the number of constraints in which it appears (i.e.,  the number of nonzero entries in column $j$ of $B$). We say an instance has \emph{arity at most $k$} if every constraint has arity at most $k$, and \emph{degree at most $D$} if every variable has degree at 
    most $D$.
\end{definition}

We assume that the acceptance sets $F_i$ are given as an explicit list. In what follows, the arity is the number of variables participating in a single constraint.

\begin{definition}[\maxlin$(q,r)$ and 
variants]\label{def:max-linsat-qr}
We write \maxlinqr{$q$}{$r$} for the restriction of \maxlin{} over $\mathbb{F}_q$ to instances in which $|F_i| = r$ for every $i$. The prefix ``\textup{E}$k$-'' restricts to instances of arity \emph{exactly} $k$, while ``$k$-'' restricts to arity \emph{at most} $k$. Thus:
\begin{itemize}
    \item \maxlinqr{$q$}{$1$} consists of systems of linear equations $L_i(\xx) = b_i$ over $\mathbb{F}_q$ (singleton acceptance sets);
    \item \maxeklinq{$k$}{$q$} is \maxlinqr{$q$}{$1$} restricted to arity exactly $k$;
    \item \maxeklinsat{$k$}{$q$}{$r$} is \maxlinqr{$q$}{$r$} 
    restricted to arity exactly $k$;
    \item \maxxor{} $= \maxlin(2,1)$ is the special case over $\mathbb{F}_2$; and
    \item \maxkxor{} restricts \maxxor{} to arity at most $k$, with \maxekxor{$k$} denoting arity exactly $k$.
\end{itemize}
In particular, \maxekxor{$k$} and \maxeklinq{$k$}{$2$} are the same problem.
\end{definition}

When every variable has degree exactly $D$, and every constraint has arity exactly $k$, the instance is called \emph{$(k,D)$-regular}; its constraint--variable bipartite graph is then $(k,D)$-biregular and satisfies $km = Dn$. We will also need the following unrelated problem, whose bounded-degree inapproximability was established by Trevisan~\cite{Trevisan2001}:

\begin{definition}[\maxeksat{$k$}]\label{def:ek-sat}
An instance of \maxeksat{$k$} consists of $n$ Boolean variables $x_1, \dots, x_n \in \{0,1\}$ and $m$ clauses, each being a disjunction of exactly $k$ literals (a literal is a variable $x_j$ or its negation $\overline{x}_j$). The goal is to find an assignment maximizing the number of satisfied clauses. 
A random assignment satisfies each clause with probability $1 - 2^{-k}$, giving a baseline of $7/8$ for $k = 3$.
\end{definition}

We explicitly emphasize that \maxeksat{$3$} and \maxeklinq{$3$}{$2$} $=$ \maxekxor{$3$} are distinct problems: in the former, each constraint is a disjunction of three literals (satisfiable by $7$ out of $8$ assignments), while in the latter each constraint is a linear equation over $\mathbb{F}_2$ (satisfiable by exactly half of all assignments). Trevisan's degree-reduction technique~\cite{Trevisan2001} has originally been
stated for \maxeksat{$3$} but applies to any CSP; we exploit this generality in 
\cref{thm:bounded_degree_linsat}.

\subsection{Inapproximability landscape}

We now summarize the inapproximability landscape (i.e., PCP-hardness statements), proceeding from unbounded-degree to bounded-degree instances.

\subsubsection{Unbounded degree}

One of the building blocks of the proof we establish is the seminal work of
Ref.~\cite{hastad2001}.
H\r{a}stad~\cite[Theorem~5.9]{hastad2001}, building on the celebrated PCP theorem~\cite{AS98,ALMSS98}:
\begin{theorem}[Inapproximability of \maxeklinq{$3$}{$\Gamma$}~\cite{hastad2001}]
\label{thm:hastad}
    For every finite Abelian group~$\Gamma$ and every $\varepsilon>0$, it is $\NP$-hard to approximate \textup{max-E}$3$\textup{-LIN-}$\Gamma$ within a factor $|\Gamma|-\varepsilon$. Equivalently, \textup{max-E}$3$\textup{-LIN-}$\Gamma$ is non-approximable beyond the random assignment threshold: it is $\NP$-hard to distinguish, given an instance of \textup{max-E}$3$\textup{-LIN-}$\Gamma$, between
    \begin{itemize}
       \item[\textup{(Y)}] instances with $\OPT \ge (1-\varepsilon)\,m$ and
       \item[\textup{(N)}] instances with $\OPT \le (1/|\Gamma|+\varepsilon)\,m$.
    \end{itemize}
    In particular, setting $\Gamma=(\mathbb{F}_q, +)$ for a finite field $\mathbb{F}_q$ yields the $(1-\varepsilon,\;1/q+\varepsilon)$-hardness of \textup{max-E}$3$\textup{-LIN-}$q$.
\end{theorem}
\begin{proposition}[Inapproximability of \maxeklinq{$k$}{$q$}]\label{rem:arity-k}
The arity-$3$ hardness of approximating \maxeklinq{$3$}{$q$}~\cite{hastad2001} extends to every $k \geq 3$: the linear predicate $\{x \in \mathbb{F}_q^k : \sum_i x_i = b\}$ is a balanced pairwise-independent coset of density $1/q$, so $(1-\varepsilon,\ 1/q+\varepsilon)$--inapproximability of \maxeklinq{$k$}{$q$} follows from Chan~\cite[Theorem~1.1]{Chan2016}.
\end{proposition}
We have recently generalized the inapproximability of \maxeklinq{$3$}{$q$} in Ref.~\cite{kramer2026tight} to \maxlinqr{$q$}{$r$} with arbitrary-sized acceptance sets, yielding the following inapproximability theorem:
\begin{theorem}[Inapproximability of \maxlinqr{$q$}{$r$}~\cite{kramer2026tight}]
\label{thm:inapproximability_of_max_linsat}
For every finite field $\mathbb{F}_q$, every integer $1\leq r \leq q-1$, and every $\varepsilon>0$, it is $\NP$-hard to distinguish, given an instance of \maxlinqr{$q$}{$r$}, between
\begin{itemize}
    \item[\textup{(Y)}] instances with $\OPT \ge (1-\varepsilon)\,m$ and
    \item[\textup{(N)}] instances with $\OPT \le (r/q+\varepsilon)\,m$.
\end{itemize}
This means it is $\NP$-hard to approximate within a factor strictly better than $r/q+\varepsilon$ for every $\varepsilon>0$.
\end{theorem}
The hard instances produced by this reduction have arity exactly $3$, i.e.,  they are \maxeklinsat{$3$}{$q$}{$r$} instances.
Note that this inapproximability guarantee is \emph{tight} since a uniformly chosen assignment $\xx \in \mathbb{F}_q$ satisfies each constraint $i$ with probability $\Abs{F_i}/q = r/q$ and thus gives an expected approximation ratio of $r/q$.
The hard instances produced by these PCP reductions have unbounded variable degree and $n/m = o(1)$. When the degree is bounded, and hence when each variable appears in  at most $D$ constraints, the approximability landscape changes qualitatively.

\subsubsection{Bounded degree}

When the degree is bounded, the inapproximability landscape changes: the random-assignment threshold is no longer a hard wall, and algorithms can provably exceed it~\cite{Hastad2000,FGG14bounded,Barak2015}. We will review these algorithms in detail in \cref{sec:algos_bounded_degree}. 
On the hardness side, the relevant result 
on which we build here is due to Trevisan~\cite{Trevisan2001}, who has 
introduced a randomized degree-reduction technique (variable splitting, weighted replication, random sampling, and pruning) that converts arbitrary CSP instances into bounded-degree instances while incurring a $\calO(1/\sqrt{D})$ loss in the approximation gap. Applied to H\r{a}stad's optimal inapproximability of \maxeksat{$3$}~\cite{hastad2001}, this yields the following theorem.

\begin{theorem}[Bounded-degree inapproximability of 
\maxeksat{$3$}~{\cite{Trevisan2001}}]\label{thm:trevisan}
For every sufficiently large degree bound $D$, it is $\NP$-hard to approximate \maxeksat{$3$} on instances where each variable appears in at most $D$ clauses within $7/8 + \calO(1/\sqrt{D})$.
\end{theorem}

Since Trevisan's reduction is randomized and succeeds with constant probability, 
the inapproximability with polynomial time classical algorithms holds under $\NP \not\subseteq \coRP$\footnote{Trevisan states the collapse as $\RP=\NP$; the reduction in fact gives $\NP\subseteq\coRP$ directly, since the YES case is preserved with probability $1$ (perfect completeness) and only the soundness side carries error.} rather than $\Pclass \neq \NP$.

Although stated originally for \maxeksat{$3$}, Trevisan's reduction is purely combinatorial and applies to any CSP over a finite alphabet~\cite{Trevisan2001,Trevisan2015blog}. Applying it instead to H\r{a}stad's $(1-\varepsilon,\; 1/2+\varepsilon)$-gap for \maxeklinq{$3$}{$2$} \cite{hastad2001} yields $\NP$-hardness of approximating bounded-degree \maxekxor{$3$} within $1/2 + \calO(1/\sqrt{D})$. This appears to be folklore, though it has not been made explicit in the literature.\footnote{The composition is referenced informally in Refs.~\cite{FGG14bounded,Trevisan2015blog}. It also follows as a special case of our \cref{thm:bounded_degree_linsat}.}
A central contribution of the present work is to extend this bounded-degree hardness from \maxxor{} to \maxlinqr{$q$}{$r$} over arbitrary finite fields $\mathbb{F}_q$ with acceptance sets of size $r$, which we carry out in \cref{sec:bounded_degree_maxlin}.

\section{Inapproximability of bounded-degree \maxlin} 
\label{sec:bounded_degree_maxlin}
The proof combines three ingredients:
\begin{enumerate}
    \item H\r{a}stad's $(1-\varepsilon,\; 1/q+\varepsilon)$--inapproximability for \maxeklinq{$3$}{$q$}~\cite{hastad2001}, which produces hard instances with unbounded variable degree.
    \item Trevisan's randomized degree-reduction technique~\cite{Trevisan2001}, which converts these into bounded-degree instances at the cost of a $\calO(1/\sqrt{D})$ loss in the gap.
    \item The deterministic reduction from \maxeklinsat{$3$}{$q$}{$1$} to \maxeklinsat{$3$}{$q$}{$r$} from Ref.~\cite{kramer2026tight}, which lifts the result to arbitrary acceptance-set size $r$ at the cost of a constant factor in the degree.
\end{enumerate}
Composing these yields the following result.
\begin{theorem}[Bounded-degree inapproximability of \maxeklinsat{$3$}{$q$}{$r$}]\label{thm:bounded_degree_linsat}
For every finite field $\mathbb{F}_q$, every integer $1 \leq r \leq q-1$, and every sufficiently large degree bound $D$, it is $\NP$-hard under randomized reductions to approximate \maxeklinsat{$3$}{$q$}{$r$} on instances where each variable appears in at most $D$ constraints, within
\begin{equation}
    \frac{r}{q} + \calO_{q,r}\!\left(\frac{1}{\sqrt{D}}\right).
\end{equation}
\end{theorem}
\begin{proof}
We compose two reductions. The first reduction adapts Trevisan's randomized degree-re\-duc\-tion technique~\cite{Trevisan2001}---stated for \maxeksat{$3$} but applicable to any CSP over a finite alphabet~\cite{Trevisan2001,Trevisan2015blog}---to convert H\r{a}stad's unbounded-degree hard instances of \maxeklinq{$3$}{$q$}~\cite{hastad2001}, i.e.,  \maxeklinsat{$3$}{$q$}{$1$}, into bounded-degree instances. 
The second reduction applies the deterministic lifting from \maxeklinsat{$3$}{$q$}{$1$} to \maxeklinsat{$3$}{$q$}{$r$} from Ref.~\cite{kramer2026tight}; since this multiplies the degree by $\binom{q-1}{r-1}$, we choose the initial degree bound to $\tilde{D} = \lfloor D / \binom{q-1}{r-1} \rfloor$ so that the final degree after lifting is at most $D$.

\begin{figure}[ht]
\begin{tikzpicture}[xscale=0.87,
  font=\tiny,
  >={Stealth[length=2mm, width=1.5mm]},
  bigvar1/.style={circle, draw=DarkOrchid!75, fill=DarkOrchid!32, line width=0.5pt,
    minimum size=6mm, inner sep=0pt, text=DarkOrchid!95!black,  font=\small},
  bigvar2/.style={circle, draw=blue!70,   fill=blue!21,    line width=0.5pt,
    minimum size=6mm, inner sep=0pt, text=blue!90,          font=\small},
  bigvar3/.style={circle, draw=teal!80,  fill=teal!30,   line width=0.5pt,
    minimum size=6mm, inner sep=0pt, text=teal!90!black,   font=\small},
  bigvar4/.style={circle, draw=Chocolate!80,  fill=Chocolate!22,   line width=0.5pt,
    minimum size=6mm, inner sep=0pt, text=IndianRed!95!black,   font=\small},
  bigvar5/.style={circle, draw=DodgerBlue,    fill=DodgerBlue!40,     line width=0.5pt,
    minimum size=6mm, inner sep=0pt, text=MidnightBlue!95!black,     font=\small},
  bigvar6/.style={circle, draw=purple!80, fill=purple!22,  line width=0.5pt,
    minimum size=6mm, inner sep=0pt, text=purple!95!black,  font=\small},
  cpvar1/.style={circle, draw=DarkOrchid!75, fill=DarkOrchid!32, line width=0.4pt,
    minimum size=4.5mm, inner sep=0pt, font=\tiny, text=DarkOrchid!95!black},
  cpvar2/.style={circle, draw=blue!70,   fill=blue!21,    line width=0.4pt,
    minimum size=4.5mm, inner sep=0pt, font=\tiny, text=blue!90},
  cpvar3/.style={circle, draw=teal!80,  fill=teal!30,   line width=0.4pt,
    minimum size=4.5mm, inner sep=0pt, font=\tiny, text=teal!90!black},
  cpvar4/.style={circle, draw=Chocolate!80,  fill=Chocolate!22,   line width=0.4pt,
    minimum size=4.5mm, inner sep=0pt, font=\tiny, text=IndianRed!95!black},
  cpvar5/.style={circle, draw=DodgerBlue,    fill=DodgerBlue!40,     line width=0.4pt,
    minimum size=4.5mm, inner sep=0pt, font=\tiny, text=MidnightBlue!95!black},
  cpvar6/.style={circle, draw=purple!80, fill=purple!22,  line width=0.4pt,
    minimum size=4.5mm, inner sep=0pt, font=\tiny, text=purple!95!black},
  frvar/.style={circle, draw=black!55, fill=none, line width=0.4pt, dashed,
    minimum size=5mm, inner sep=0pt, font=\tiny},
  bigcon/.style={rectangle, rounded corners=0.5pt, draw=gray!70, fill=LightSteelBlue!50,
    line width=0.5pt, minimum size=5mm, inner sep=1pt, text=gray!85, font=\scriptsize},
  cpcon/.style={rectangle, rounded corners=0.3pt, draw=gray!70, fill=LightSteelBlue!50,
    line width=0.3pt, minimum size=2.5mm, inner sep=0pt},
  cpconH/.style={rectangle, rounded corners=0.3pt, draw=black!35, fill=none,
    line width=0.25pt, minimum size=2.5mm, inner sep=0pt},
  cpconP/.style={rectangle, rounded corners=0.3pt, draw=orange!85, fill=orange!28,
    line width=0.4pt, minimum size=3mm, inner sep=0pt},
  cpconR/.style={rectangle, rounded corners=0.3pt, draw=black!55, fill=none,
    line width=0.4pt, dashed, minimum size=3mm, inner sep=0pt},
  edg/.style={line width=0.3pt, black!55},
  edgD/.style={line width=0.3pt, black!50, dashed},
  panel/.style={rectangle, rounded corners=2pt, draw=black!25, line width=0.4pt},
  ptitle/.style={font=\large},
  psub/.style={font=\scriptsize, text=black!65},
  psubW/.style={font=\scriptsize, text=black!65, fill=white, inner sep=1pt},
  arlbl/.style={font=\scriptsize, text=black!75},
]

\draw[panel] (0,0.8) rectangle (2.45,8.9);
\node[ptitle] at (1.25,9.5) {$\phi$};
\node[psub] at (1.25,9.15) {$n=6,\ m=5$};

\node[bigvar1] (xA1) at (0.55,7.8) {$x_1$};
\node[bigvar2] (xA2) at (0.55,6.8) {$x_2$};
\node[bigvar3] (xA3) at (0.55,5.8) {$x_3$};
\node[bigvar4] (xA4) at (0.55,4.2) {$x_4$};
\node[bigvar5] (xA5) at (0.55,3.2) {$x_5$};
\node[bigvar6] (xA6) at (0.55,2.2) {$x_6$};

\begin{scope}[shift={(0,0.45)}]
\node[bigcon] (CA1) at (2.0,7.5) {$C_1$};
\node[bigcon] (CA2) at (2.0,6.0) {$C_2$};
\node[bigcon] (CA3) at (2.0,4.5) {$C_3$};
\node[bigcon] (CA4) at (2.0,3.0) {$C_4$};
\node[bigcon] (CA5) at (2.0,1.5) {$C_5$};
\end{scope}

\draw[edg] (xA1) -- (CA1); \draw[edg] (xA2) -- (CA1); \draw[edg] (xA3) -- (CA1);
\draw[edg] (xA1) -- (CA2); \draw[edg] (xA2) -- (CA2); \draw[edg] (xA4) -- (CA2);
\draw[edg] (xA1) -- (CA3); \draw[edg] (xA3) -- (CA3); \draw[edg] (xA5) -- (CA3);
\draw[edg] (xA2) -- (CA4); \draw[edg] (xA3) -- (CA4); \draw[edg] (xA6) -- (CA4);
\draw[edg] (xA4) -- (CA5); \draw[edg] (xA5) -- (CA5); \draw[edg] (xA6) -- (CA5);

\draw[->] (2.65,4.5) -- (3.4,4.5);
\node[arlbl] at (3.05,5.0) {split \&};
\node[arlbl] at (3.05,4.75) {expand};

\draw[panel] (3.65,0.8) rectangle (8.3,8.9);
\node[ptitle] at (6.0,9.5) {$\phi_w$};
\node[psub] at (6.0,9.15) {$89$ constraint expansions};

\node[cpvar1] (yB1_1) at (4.0,8.4)  {$y_1^1$};
\node[cpvar1] (yB1_2) at (4.0,7.95) {$y_1^2$};
\node[cpvar1] (yB1_3) at (4.0,7.5)  {$y_1^3$};
\node[cpvar2] (yB2_1) at (4.0,6.9)  {$y_2^1$};
\node[cpvar2] (yB2_2) at (4.0,6.45) {$y_2^2$};
\node[cpvar2] (yB2_3) at (4.0,6.0)  {$y_2^3$};
\node[cpvar3] (yB3_1) at (4.0,5.4)  {$y_3^1$};
\node[cpvar3] (yB3_2) at (4.0,4.95) {$y_3^2$};
\node[cpvar3] (yB3_3) at (4.0,4.5)  {$y_3^3$};
\node[cpvar4] (yB4_1) at (4.0,3.9)  {$y_4^1$};
\node[cpvar4] (yB4_2) at (4.0,3.45) {$y_4^2$};
\node[cpvar5] (yB5_1) at (4.0,2.85) {$y_5^1$};
\node[cpvar5] (yB5_2) at (4.0,2.4)  {$y_5^2$};
\node[cpvar6] (yB6_1) at (4.0,1.8)  {$y_6^1$};
\node[cpvar6] (yB6_2) at (4.0,1.35) {$y_6^2$};

\foreach \i in {0,...,8} {\foreach \j in {0,1,2} {\node[cpconH] at ({5.3+0.3*\i},{8.4-0.3*\j}) {};}}
\foreach \i in {0,...,5} {\foreach \j in {0,1,2} {\node[cpconH] at ({5.3+0.3*\i},{6.9-0.3*\j}) {};}}
\foreach \i in {0,...,5} {\foreach \j in {0,1,2} {\node[cpconH] at ({5.3+0.3*\i},{5.4-0.3*\j}) {};}}
\foreach \i in {0,...,5} {\foreach \j in {0,1,2} {\node[cpconH] at ({5.3+0.3*\i},{3.9-0.3*\j}) {};}}
\foreach \i in {0,...,3} {\foreach \j in {0,1}   {\node[cpconH] at ({5.3+0.3*\i},{2.4-0.3*\j}) {};}}

\node[cpcon] (Bs1) at (5.3, 8.4) {};
\node[cpcon] (Bs2) at (6.2, 8.1) {};
\node[cpcon] (Bs3) at (7.4, 7.8) {};
\node[cpcon] (Bs4) at (5.6, 6.9) {};
\node[cpcon] (Bs5) at (6.5, 6.3) {};
\node[cpcon] (Bs6) at (6.2, 5.1) {};
\node[cpcon] (Bs7) at (5.9, 3.6) {};

\draw[edg] (yB1_1)--(Bs1); \draw[edg] (yB2_1)--(Bs1); \draw[edg] (yB3_1)--(Bs1);
\draw[edg] (yB1_1)--(Bs2); \draw[edg] (yB2_2)--(Bs2); \draw[edg] (yB3_2)--(Bs2);
\draw[edg] (yB1_2)--(Bs3); \draw[edg] (yB2_3)--(Bs3); \draw[edg] (yB3_3)--(Bs3);
\draw[edg] (yB1_1)--(Bs4); \draw[edg] (yB2_1)--(Bs4); \draw[edg] (yB4_1)--(Bs4);
\draw[edg] (yB1_3)--(Bs5); \draw[edg] (yB2_2)--(Bs5); \draw[edg] (yB4_2)--(Bs5);
\draw[edg] (yB1_1)--(Bs6); \draw[edg] (yB3_3)--(Bs6); \draw[edg] (yB5_1)--(Bs6);
\draw[edg] (yB2_3)--(Bs7); \draw[edg] (yB3_2)--(Bs7); \draw[edg] (yB6_1)--(Bs7);

\begin{scope}[shift={(0,-0.05)}]
\node[psubW] at (6.5, 7.5)  {$C_1$ (27 expansions)};
\node[psubW] at (6.5, 6.0) {$C_2$ (18 expansions)};
\node[psubW] at (6.5, 4.5) {$C_3$ (18 expansions)};
\node[psubW] at (6.5, 3.0) {$C_4$ (18 expansions)};
\node[psubW] at (6.5, 1.8) {$C_5$ (8 expansions)};
\end{scope}

\draw[->] (8.5,4.5) -- (9.25,4.5);
\node[arlbl] at (8.9,4.75) {sample};

\draw[panel] (9.5,0.8) rectangle (12.5,8.9);
\node[ptitle] at (11.0,9.5) {$\phi_R$};
\node[psub] at (11.0,9.15) {$7$ samples, pre-prune};

\node[cpvar1] (yC1_1) at (10.0,8.4)  {$y_1^1$};
\node[cpvar1] (yC1_2) at (10.0,7.95) {$y_1^2$};
\node[cpvar1] (yC1_3) at (10.0,7.5)  {$y_1^3$};
\node[cpvar2] (yC2_1) at (10.0,6.9)  {$y_2^1$};
\node[cpvar2] (yC2_2) at (10.0,6.45) {$y_2^2$};
\node[cpvar2] (yC2_3) at (10.0,6.0)  {$y_2^3$};
\node[cpvar3] (yC3_1) at (10.0,5.4)  {$y_3^1$};
\node[cpvar3] (yC3_2) at (10.0,4.95) {$y_3^2$};
\node[cpvar3] (yC3_3) at (10.0,4.5)  {$y_3^3$};
\node[cpvar4] (yC4_1) at (10.0,3.9)  {$y_4^1$};
\node[cpvar4] (yC4_2) at (10.0,3.45) {$y_4^2$};
\node[cpvar5] (yC5_1) at (10.0,2.85) {$y_5^1$};
\node[cpvar5] (yC5_2) at (10.0,2.4)  {$y_5^2$};
\node[cpvar6] (yC6_1) at (10.0,1.8)  {$y_6^1$};
\node[cpvar6] (yC6_2) at (10.0,1.35) {$y_6^2$};

\node[cpconP] (Cc1) at (12.0,8.0) {};
\node[cpcon]  (Cc2) at (12.0,7.1) {};
\node[cpcon]  (Cc3) at (12.0,6.2) {};
\node[cpcon]  (Cc4) at (12.0,5.3) {};
\node[cpcon]  (Cc5) at (12.0,4.4) {};
\node[cpcon] (Cc6) at (12.0,3.5) {};
\node[cpcon]  (Cc7) at (12.0,2.6) {};

\draw[edg] (yC1_1)--(Cc1); \draw[edg] (yC2_1)--(Cc1); \draw[edg] (yC3_1)--(Cc1);
\draw[edg] (yC1_1)--(Cc2); \draw[edg] (yC2_2)--(Cc2); \draw[edg] (yC3_2)--(Cc2);
\draw[edg] (yC1_2)--(Cc3); \draw[edg] (yC2_3)--(Cc3); \draw[edg] (yC3_3)--(Cc3);
\draw[edg] (yC1_1)--(Cc4); \draw[edg] (yC2_1)--(Cc4); \draw[edg] (yC4_1)--(Cc4);
\draw[edg] (yC1_3)--(Cc5); \draw[edg] (yC2_2)--(Cc5); \draw[edg] (yC4_2)--(Cc5);
\draw[edg] (yC1_1)--(Cc6); \draw[edg] (yC3_3)--(Cc6); \draw[edg] (yC5_1)--(Cc6);
\draw[edg] (yC2_3)--(Cc7); \draw[edg] (yC3_2)--(Cc7); \draw[edg] (yC6_1)--(Cc7);

\draw[->] (12.65,4.5) -- (13.35,4.5);
\node[arlbl] at (13.0,4.75) {prune};

\draw[panel] (13.5,0.8) rectangle (17.0,8.9);
\node[ptitle] at (15.25,9.5) {$\phi_{\tilde D}$};
\node[psub] at (15.25,9.15) {degree $\le \tilde D = 3$};

\node[cpvar1] (yD1_1) at (14.0,8.4)  {$y_1^1$};
\node[cpvar1] (yD1_2) at (14.0,7.95) {$y_1^2$};
\node[cpvar1] (yD1_3) at (14.0,7.5)  {$y_1^3$};
\node[cpvar2] (yD2_1) at (14.0,6.9)  {$y_2^1$};
\node[cpvar2] (yD2_2) at (14.0,6.45) {$y_2^2$};
\node[cpvar2] (yD2_3) at (14.0,6.0)  {$y_2^3$};
\node[cpvar3] (yD3_1) at (14.0,5.4)  {$y_3^1$};
\node[cpvar3] (yD3_2) at (14.0,4.95) {$y_3^2$};
\node[cpvar3] (yD3_3) at (14.0,4.5)  {$y_3^3$};
\node[cpvar4] (yD4_1) at (14.0,3.9)  {$y_4^1$};
\node[cpvar4] (yD4_2) at (14.0,3.45) {$y_4^2$};
\node[cpvar5] (yD5_1) at (14.0,2.85) {$y_5^1$};
\node[cpvar5] (yD5_2) at (14.0,2.4)  {$y_5^2$};
\node[cpvar6] (yD6_1) at (14.0,1.8)  {$y_6^1$};
\node[cpvar6] (yD6_2) at (14.0,1.35) {$y_6^2$};

\node[cpconP] (Cd1) at (16.0,8.0) {};
\node[cpcon]  (Cd2) at (16.0,7.1) {};
\node[cpcon]  (Cd3) at (16.0,6.2) {};
\node[cpcon]  (Cd4) at (16.0,5.3) {};
\node[cpcon]  (Cd5) at (16.0,4.4) {};
\node[cpcon] (Cd6) at (16.0,3.5) {};
\node[cpcon]  (Cd7) at (16.0,2.6) {};

\node[frvar] (ys1) at (16.7,8.5) {$y_*^1$};
\node[frvar] (ys2) at (16.7,8.0) {$y_*^2$};
\node[frvar] (ys3) at (16.7,7.5) {$y_*^3$};


\draw[edg] (yD1_1)--(Cd2); \draw[edg] (yD2_2)--(Cd2); \draw[edg] (yD3_2)--(Cd2);
\draw[edg] (yD1_2)--(Cd3); \draw[edg] (yD2_3)--(Cd3); \draw[edg] (yD3_3)--(Cd3);
\draw[edg] (yD1_1)--(Cd4); \draw[edg] (yD2_1)--(Cd4); \draw[edg] (yD4_1)--(Cd4);
\draw[edg] (yD1_3)--(Cd5); \draw[edg] (yD2_2)--(Cd5); \draw[edg] (yD4_2)--(Cd5);
\draw[edg] (yD1_1)--(Cd6); \draw[edg] (yD3_3)--(Cd6); \draw[edg] (yD5_1)--(Cd6);
\draw[edg] (yD2_3)--(Cd7); \draw[edg] (yD3_2)--(Cd7); \draw[edg] (yD6_1)--(Cd7);
\draw[edgD] (Cd1)--(ys1); \draw[edgD] (Cd1)--(ys2); \draw[edgD] (Cd1)--(ys3);

\end{tikzpicture}

\caption{Degree reduction $\phi \to \phi_w \to \phi_R \to \phi_{\tilde D}$ on an example with $n = 6$, $m = 5$, arity $k = 3$, with variable degree vector $(o_1,\ldots,o_6) = (3,3,3,2,2,2)$ and degree bound $\tilde D = 3$. Each panel transition represents a larger step of the construction. The first panel visualizes the input instance as a graph, showing which variables belong to which constraints.
In the second panel, the weighted instance $\phi_w$ is shown: Each variable $x_i$ is copied $o_i$ times and each original constraint is expanded to a \emph{constraint group}, such that each combination of copy-variables relating to the corresponding original constraint is covered by one element of the group. These group elements are also each weighted in inverse proportion to their group's size.  
The third panel introduces the probabilistic step, where $\tilde{D}N/e^2$ constraints (rounded up) are chosen randomly in proportion to their weight.
In the last panel, constraints are iteratively pruned as long as they contain variables that have more than $\tilde D$ appearances (marked orange). This is done by replacing them with trivial new constraints on new variables $y_*^j$, thus restoring the degree cap.}
\label{fig:degree-reduction-color}
\end{figure}
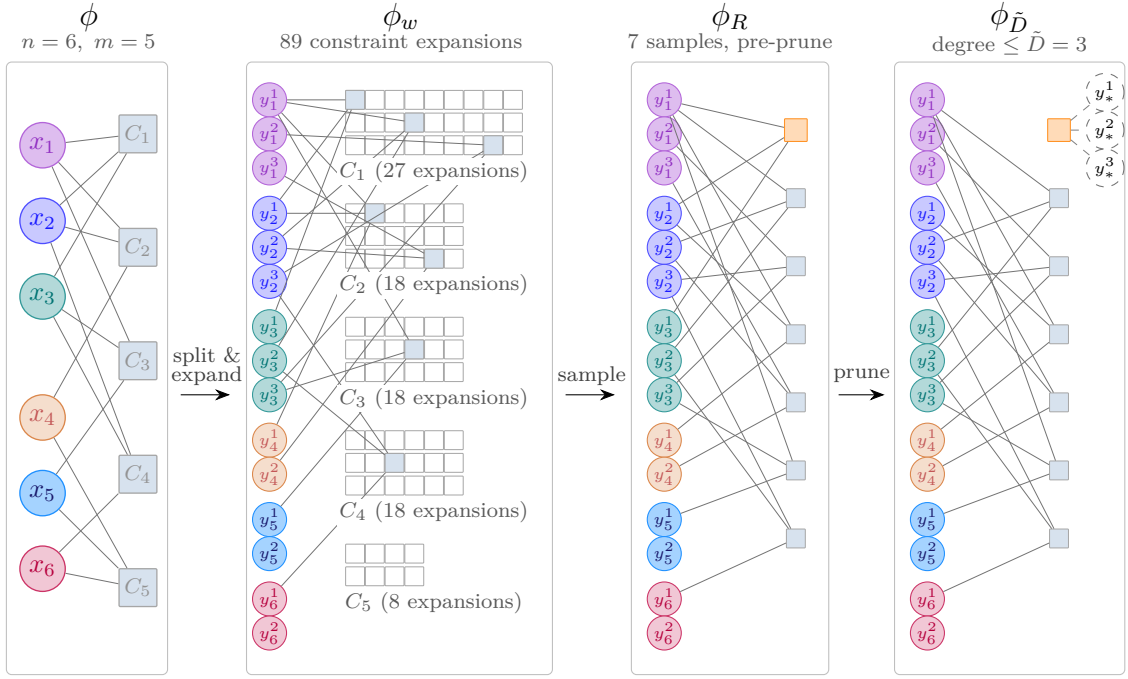

\noindent\textit{Degree reduction.} 
We give the argument in full for arbitrary $q$, since the generalization from Boolean to $\mathbb{F}_q$ alphabets has not previously appeared in the literature; the key modification is in the union bound (which ranges over $q^N$ rather than $2^N$ assignments) and the recovery step (which uses $q$-ary plurality rather than majority decoding).
Let $\phi$ be a hard instance of \maxeklinq{$3$}{$q$} with $n$ variables over $\mathbb{F}_q$, $m$ constraints, and gap $(1-\varepsilon,\; 1/q+\varepsilon)$ as produced by H\r{a}stad's PCP reduction~\cite{hastad2001}. Let $o_i$ denote the number of constraints containing $x_i$ and set $N = \sum_i o_i = 3m$. Following Trevisan~\cite{Trevisan2001}, we construct a degree-$\tilde{D}$ instance $\phi_{\tilde{D}}$ by: (i)~introducing, for each variable $x_i$, a set of $o_i$ copy-variables $y_i^1, \ldots, y_i^{o_i} \in \mathbb{F}_q$; 
(ii)~building a weighted instance $\phi_w$: for each original constraint involving variables $\{x_a, x_b, x_c\}$, we create all
constraints obtained by substituting the original variables with their respective copies, i.e.,  we build each new constraint with $\{\{y_a^{j_1}, y_b^{j_2}, y_c^{j_3}\} \mid(j_1,j_2,j_3) \in [o_a] \times [o_b] \times [o_c]\}$ (carrying over their values); each of those new constraints is then weighted by $1/(o_a o_b o_c)$;
(iii)~sampling $(\tilde{D}/e^2) \cdot N$ constraints from $\phi_w$ proportionally to weight, yielding $\phi_R$; and (iv)~pruning variables with more than $\tilde{D}$ occurrences by replacing excess constraints with trivially satisfiable ones on fresh variables, yielding $\phi_{\tilde{D}}$. An example of this construction is shown in \Cref{fig:degree-reduction-color}.

Completeness of the reduction is immediate: any assignment satisfying a $(1-\varepsilon)$--fraction of $\phi$ extends to $\phi_{\tilde{D}}$ by setting all copies equal.
We analyze soundness in two steps: First, we consider only $\phi_R$, where concentration applies, then we account for the additional error introduced by pruning in the transition from $\phi_R$ to $\phi_{\tilde{D}}$.
To begin, note that the sum of all weights in $\phi_w$ is $m$.
We fix an assignment $\mathbf{a}$ to the copy-variables and let
$\rho(\mathbf{a}) \in [0,1]$ be the sum of weights of all constraints in $\phi_w$ that are satisfied by $\mathbf{a}$, divided by $m$.
Each of the constraints in $\phi_R$ is drawn independently and proportionally to its weight so the indicator of whether $\mathbf{a}$ satisfies any one is i.i.d.\@ $\mathrm{Bernoulli}(\rho(\mathbf{a}))$. Applying Hoeff\-ding's inequality to the total number $S$ of those constraints satisfied by $\mathbf{a}$, we get
\begin{align}
\Pr\left[\frac{S}{\tilde{D}N/e^2} \ge \rho(\mathbf{a}) + \varepsilon\right]
\le \exp(-2\varepsilon^2 
\tilde{D}N/e^2).
\end{align}
Using the union bound over all possible assignments of the $N$ copy-variables, we get 
\begin{align}
    q^N \cdot\,\exp(-2\varepsilon^2 \tilde{D}N/e^2) < 1 \Longleftrightarrow 
\varepsilon > e \sqrt{(\ln q) / (2\tilde{D})},
\end{align}
so setting $\varepsilon = \Theta_q(1/\sqrt{\tilde{D}})$
ensures that simultaneously for all assignments $\mathbf{a}$, no assignment satisfies more than a
$\rho(\mathbf{a}) + \varepsilon$ fraction of the constraints in $\phi_R$ (with high probability).
We next argue that no dominant additional error is introduced in the pruning step.
For each copy-variable $y$, let $X_y$ count the constraints of $\phi_R$ that contain $y$; as each sample contains $y$ with probability $1/m$, we have $X_y \sim \mathrm{Binomial}(\tilde{D}N/e^2, 1/m)$ with mean $\mu = 3\tilde D/e^2$ (using $N=3m$) and variance $\sigma^2 \le \mu$. 
Since 
$\max\{X_y - \tilde{D}, 0\}$, called \emph{excess} of $X_y$, can be at most $\frac{(X_y - \mu)^2}{\tilde D - \mu}$ algebraically\footnote{For positive numbers $a=X_y - \tilde{D}>0$ and $b=\tilde D - \mu>0$: $(a+b)^2/b = a^2/b + 2a + b \geq a$.}, 
the expected excess is $\sigma^2/(\tilde D - \mu) \le \frac{1}{(e^2/3-1)}$.
Summing over all copy-variables, the expected number $R$ of pruned constraints satisfies $\mathbb{E}[R] \leq N/(\frac{e^2}{3}-1)$. By Markov's inequality, $R \le 4N/(\frac{e^2}{3} - 1)$ with probability at least $3/4$. The constraint replacements then contribute at most $\frac{R}{N\tilde{D}/e^2} = (4e^2/(\frac{e^2}{3}-1))/\tilde D = \calO(1/\tilde D)$ to the satisfaction fraction.

To recover an assignment $\mathbf{x}(\mathbf{a}) \in \mathbb{F}_q^n$ from $\mathbf{a}$, set $x_i = v$ with probability $|\{j : y_i^j = v\}|/o_i$. For any constraint $x_a + x_b + x_c = b$, the satisfaction probability under this recovery equals exactly the weighted fraction $W$ of its expansions satisfied by $\mathbf{a}$ in $\phi_w$:
\begin{equation}
    \Pr[x_a + x_b + x_c = b]
    = \sum_{\substack{v_1+v_2+v_3 = b}}
      \frac{|\{j : y_a^j = v_1\}|}{o_a} \cdot
      \frac{|\{j : y_b^j = v_2\}|}{o_b} \cdot
      \frac{|\{j : y_c^j = v_3\}|}{o_c}
    = W.
\end{equation}
This identity follows from the independence of the induced product distribution and holds for linear constraints over any finite Abelian group. Summing over all $m$ constraints, the expected satisfaction is $m \cdot \rho(\mathbf{a})$, so an assignment achieving at least this exists. The degree reduction thus yields \maxeklinq{$3$}{$q$} instances of degree at most $\tilde{D}$ with gap $(1-\varepsilon,\; 1/q + \calO_q (1/\sqrt{\tilde{D}} ) )$.

\noindent\textit{Lifting to \maxeklinsat{$3$}{$q$}{$r$}.}
The reduction from \maxeklinsat{$3$}{$q$}{$1$} to \maxeklinsat{$3$}{$q$}{$r$} in Ref.~\cite{kramer2026tight} replaces each constraint by $\binom{q-1}{r-1}$ constraints with $r$-element acceptance sets, multiplying the degree by $\binom{q-1}{r-1}$. With $\tilde{D} = \lfloor D/\binom{q-1}{r-1} \rfloor$, the output has degree at most $D$. The soundness maps $1/q + \delta$ to $r/q + \delta \cdot (q-r)/(q-1) \leq r/q + \delta$. Substituting $\delta = \calO_q(1/\sqrt{\tilde{D}}) = \calO_{q,r}(1/\sqrt{D})$ gives \maxeklinsat{$3$}{$q$}{$r$} instances of degree at most $D$ with gap $(1-\varepsilon,\; r/q + \calO_{q,r}(1/\sqrt{D}))$.
\end{proof}

\begin{corollary}[Bounded-degree inapproximability of \maxlinqr{$q$}{$r$}]\label{cor:bounded_degree_general}
For every finite field $\mathbb{F}_q$, every integer $1 \leq r \leq q-1$, and every
sufficiently large degree bound $D$, it is $\NP$-hard under randomized reductions to
approximate \maxlinqr{$q$}{$r$} on instances of degree at most $D$ within
$r/q + \calO_{q,r}(1/\sqrt{D})$.
\end{corollary}
\begin{proof}
\maxeklinsat{$3$}{$q$}{$r$} is the arity-exactly-$3$ restriction of \maxlinqr{$q$}{$r$}, so
the hardness of \Cref{thm:bounded_degree_linsat} transfers immediately.
\end{proof}

\begin{corollary}[Hardness for higher arity]\label{cor:arity}
For every finite field $\mathbb{F}_q$, every $1 \leq r \leq q-1$, every $k \geq 3$, and every sufficiently large degree bound $D$, it is $\NP$-hard under randomized reductions to approximate \maxeklinsat{$k$}{$q$}{$r$}
on instances of degree at most $D$ within $r/q + \calO_{q,r}(1/\sqrt{D})$.
\end{corollary}
\begin{proof}
The reduction of \Cref{thm:bounded_degree_linsat} applies verbatim to any arity $k\geq 3$, since both Trevisan's subsampling and the acceptance-set lift of Ref.~\cite{kramer2026tight} act only on whole constraints and never change which coefficients are nonzero, so they preserve arity exactly.
Thus, the reduction can also be used on hard input instances of \maxeklinq{$k$}{$q$} with arity exactly $k$ (see \cref{rem:arity-k}), yielding \maxeklinsat{$k$}{$q$}{$r$} instances of the same exact arity $k$, degree at most $D$, and the same gap.
\end{proof}

\begin{remark}[Imperfect completeness and the complexity assumption]
    Despite using the same degree-reduction technique as Trevisan, our underlying complexity assumption is $\NP \not\subseteq \BPP$ (and, analogously, $\NP \not\subseteq \BQP$ if quantum algorithms are used subsequently) rather than $\NP \not\subseteq \coRP$.
    The difference is rooted not in the reduction but in the source hardness: Trevisan starts from H\r{a}stad's \emph{perfectly complete} gap for \maxeksat{$3$} (a \emph{satisfiable} instance versus one of value $\le 7/8 + \varepsilon$), whereas our source, H\r{a}stad's $(1-\varepsilon,\ 1/q+\varepsilon)$-gap for \maxeklinq{$3$}{$q$} (\cref{thm:hastad}), has \emph{imperfect completeness}.
    The only randomized step in either reduction is the constraint-subsampling step~(iii); all other steps of the chain are deterministic.
    With perfect completeness, a satisfying assignment is preserved by subsampling with probability $1$, so Trevisan's reduction errs only on the soundness side and yields $\NP \subseteq \coRP$. With imperfect completeness, both error directions occur, each with exponentially small probability:
    \begin{itemize}
        \item \emph{Soundness direction.}
        For an input NO-instance, every assignment satisfies at most an $(r/q + \varepsilon)$--fraction of constraints.
        Fixing an optimal assignment, subsampling can over-represent the constraints it satisfies, pushing its empirical value above the YES-threshold $(1-\varepsilon)$ with very low but non-zero probability.
        The subsequent pruning afterwards can then only increase the satisfied fraction further.
        \item \emph{Completeness direction.} An optimal assignment to a YES-instance only satisfies a $(1-\varepsilon)$--fraction of constraints, leaving an $\varepsilon$--fraction unsatisfied.
        If subsampling over-represents the unsatisfied constraints of every assignment simultaneously---which may again occur with low but non-zero probability---then the reduction can no longer certify its output as a YES-instance.
    \end{itemize}
    Both directions therefore contribute (exponentially small) error, placing the reduction in the two-sided regime and yielding $\NP \subseteq \BPP$ rather than $\NP \subseteq \coRP$. In Trevisan's setting the completeness-side error vanishes ($\varepsilon = 0$ on the YES side), which accounts for the discrepancy.
    Consequently, assuming $\NP \not\subseteq \BPP$ (respectively $\NP \not\subseteq \BQP$), no classical (respectively quantum) polynomial-time algorithm can approximate \maxeklinsat{$k$}{$q$}{$r$} of bounded degree $D$ beyond $r/q + \calO_{q,r}(1/\sqrt{D})$ on worst-case instances for $k\geq 3$.
\end{remark}

\begin{remark}[Limiting cases of \cref{thm:bounded_degree_linsat}]\label{rem:limiting_cases}
The hard instances produced by the reduction have arity exactly $3$. Setting $q = 2$ and $r = 1$ recovers the claimed hardness for bounded-degree \maxekxor{$3$}.
In the limit $D \to \infty$, the $\calO_{q,r}(1/\sqrt{D})$ term vanishes, and the bound collapses to the unbounded-degree wall $r/q$ of \cref{thm:inapproximability_of_max_linsat}, consistent with the fact that H\r{a}stad's PCP reduction~\cite{hastad2001} produces instances with unbounded variable degree. \Cref{thm:bounded_degree_linsat} thus unifies the bounded- and unbounded-degree settings.
\end{remark}

\section{Quantum and classical algorithms for bounded degree instances}
\label{sec:algos_bounded_degree}

\begin{figure}[t]
    \centering
    \includegraphics[width=0.9\linewidth]{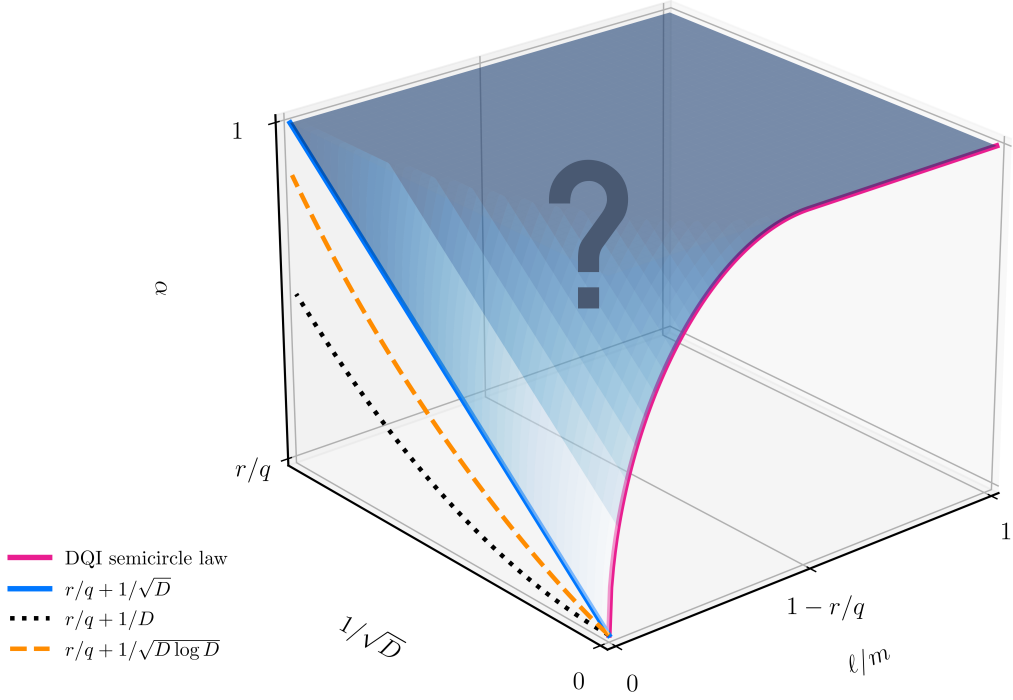}
    \caption{
    Approximability landscape for bounded-degree \maxeklinsat{$k$}{$q$}{$r$} as a function of $1/\sqrt{D}$ and normalized decoding radius $\ell/m$ of the dual code used in DQI~\cite{Jordan2024DQI}. The vertical axis is the approximation ratio. The plane interpolates between (i) worst-case bounded-degree instances with no decodable structure ($\ell/m=0$) and (ii) structured instances with decoding gain ($\ell/m>0$). For $\ell/m=0$, the blue curve shows the worst-case hardness bound $r/q + \mathcal{O}(1/\sqrt{D})$ from \cref{cor:arity}, tight up to constants due to Barak et al.~\cite{Barak2015} (for $r=1$ and $q=2$). The orange and black curves indicate representative scalings of algorithms. For fixed $D$, the magenta curve shows the DQI semicircle law exploiting decodable structure. The hard instances produced by the reduction have $\ell/m = 0$, while the DQI regime reflects structured instances with $\ell/m > 0$.
    Characterizing the tight approximability across the full $(D, \ell/m)$ plane remains an important open problem. 
    }
    \label{fig:sketch}
\end{figure}

The resulting landscape for \maxlinqr{$q$}{$r$} is also graphi\-cally summarized in \cref{fig:sketch}. We now complement the hardness results with the algorithmic landscape for bounded-degree \maxlinqr{$q$}{$r$}. The statements and algorithms for the special case of \maxekxor{$k$} are summarized in \cref{tab:constants}.

\begin{table}[ht]
\centering
\small
\setlength{\aboverulesep}{0pt}
\setlength{\belowrulesep}{0pt}
\setlength{\abovetopsep}{0pt}
\setlength{\belowbottomsep}{0pt}
\setlength{\tabcolsep}{5pt}
\begin{tabular}{@{}lp{5cm}c@{}}

\toprule
Algorithm & Advantage over $1/2$ & Q/C \\
\midrule

\multicolumn{3}{@{}l}{\textbf{Worst-case guarantees}} \\

H\r{a}stad greedy~\cite{Hastad2000}
& $\Omega(1/D)$ 
& C \\

Depth-$1$ QAOA~\cite{FGG14bounded}
& $c/(\sqrt{D}\ln D)$, $c =1/101$
& Q \\

Barak et al.~\cite{Barak2015} (odd $k$ or triangle-free instances)
& $\exp(-\calO(k))/\sqrt{D}$
& C \\

\midrule

\multicolumn{3}{@{}l}{\textbf{Complexity-theoretic limitations}} \\

$\NP$-hardness (cf.~\cite{Trevisan2001,Trevisan2015blog} and \cref{cor:arity})
& optimal worst-case scaling: $\Theta_k(1/\sqrt{D})$ (up to constants) 
& -- \\

\midrule

\multicolumn{3}{@{}l}{\textbf{Information-theoretic upper bound}} \\

DQI with classical decoders\newline
(\cref{rem:dqi_ceiling})
& $\calO(1/\sqrt{D \log D})$
& Q \\

DQI with quantum decoders\newline 
(\cref{rem:dqi_ceiling})
& $\calO(1/\sqrt{D})$
& Q \\

\bottomrule
\end{tabular}

\caption{
Approximation landscape for bounded-degree \maxekxor{$k$} at large degree $D$. The table summarizes worst-case algorithmic guarantees, complexity-theoretic hardness, and information-theoretic ceilings for DQI on $(k,D)$-regular instances. The hardness results and the algorithm of Barak et al.\ together show that the optimal worst-case scaling in $D$ is $\Theta(1/\sqrt{D})$ up to constant factors.
}
\label{tab:constants}
\end{table}

\subsection{Bounded-degree \maxxor}

\paragraph{Worst-case algorithms and optimal scaling.}
H\r{a}stad~\cite{Hastad2000} has shown first that bounded-degree CSP instances are not approximation-resistant: For any CSP with degree at most $D$, a greedy algorithm achieves $\mu + \Omega(1/D)$ for $\mu$ being the random assignment threshold. For \maxekxor{$k$} this gives $1/2 + \Omega(1/D)$. Depth-$1$ QAOA improved the advantage to $\Omega(D^{-3/4})$~\cite{FGG14bounded}. Barak et al.~\cite{Barak2015} (see also Ref.~\cite{hastad2015improved}) achieved the optimal scaling classically: for \maxekxor{$k$} with odd $k$ and degree bound $D$, their algorithm finds an assignment satisfying at least a $1/2 + \exp(-\calO(k))/\sqrt{D}$ fraction of constraints.\footnote{A subsequent version of Ref.~\cite{FGG14bounded} independently obtained $1/2 + 1/(101\sqrt{D}\ln D)$ and $1/2 + \mathrm{const}/\sqrt{D}$ on ``typical'' instances.} Combining with the bounded-degree hardness from \cref{cor:arity} pins down the optimal worst-case scaling:

\begin{observation}[Optimal scaling for bounded-degree 
\maxekxor{$k$}]\label{obs:scaling}
For \maxekxor{$k$} with fixed odd $k$, degree at most $D$, and $D$ sufficiently large, the optimal polynomial-time approximation ratio on worst-case instances 
is $1/2 + \Theta_k(1/\sqrt{D})$: the hardness bound $1/2 + \calO(1/\sqrt{D})$ follows from \cref{cor:arity} (assuming $\NP \not\subseteq \BPP$), and the matching algorithm achieving $1/2 + \exp(-\calO(k))/\sqrt{D}$ is due to Barak et al.~\cite{Barak2015}. For even $k$, the algorithmic bound requires an additional triangle-freeness assumption.
\end{observation}

Despite extensive recent work on bounded-degree \maxkxor{} via DQI and QAOA, this inapproximability landscape appears not to have been explicitly connected to these algorithms. In light of \cref{cor:arity}, any improvement---classical or quantum---is confined to the constant prefactor $c_k$ in $1/2 + c_k/\sqrt{D}$. The worst-case constant $\exp(-\calO(k))$ from Barak et al.\ is not computed explicitly and is likely very small for moderate $k$.

\paragraph{Average-case constants on Gallager instances.}
On random $(k,D)$-regular instances drawn from Gallager's ensemble, a richer picture
emerges; we summarize the landscape here, following the systematic catalog of
Ref.~\cite{shutty2026optimizationusinglocallyquantumdecoders}. Prange's infor\-ma\-tion-set decoding~\cite{Prange} solves a random $n \times n$ subsystem exactly and satisfies the residual constraints in expectation, giving $1/2 + k/(2D)$.
All other average-case algorithms achieve the $1/\sqrt{D}$ improvement. Ref.~\cite{shutty2026optimizationusinglocallyquantumdecoders} establishes that in the large $D$ limit, 
both Regev/DQI with locally-quantum decoders (Regev+FGUM) and the classical Turbo Prange heuristic achieve $1/2 + 1/\sqrt{2\pi D}$, arising from the central limit theorem applied to binomial deviations. Remarkably, the relatively simple FGUM decoder thus already matches the $1/\sqrt{D}$ scaling and the exact constant $1/\sqrt{2\pi}$ of the best classical heuristic.
Notably, the upgrade from Prange's order $1/D$-improvement to Turbo Prange's order $1/\sqrt{D}$-improvement comes from a single modification: a greedy local improvement step. 
QAOA at depth $p$ achieves $1/2 + \nu_{p}^{[k]}\sqrt{k/(2(D{-}1))}$ with numerically calculated values $\nu_p^{[k]}$ for $k$ up to $6$ given in Ref.~\cite{Basso2021}.
A systematic comparison of the constants achieved by these algorithms, including DQI with various decoders, QAOA at different depths, and classical heuristics, is provided in Ref.~\cite{shutty2026optimizationusinglocallyquantumdecoders}. Additionally, Piveteau~\cite{Piveteau2026} computes the constants achieved by combining DQI with \emph{belief propagation with quantum messages}  (BPQM)~\cite{Renes2017,BMP2022,mandal2026beliefpropagationquantummessages} for different combinations of $(k,D)$.

In a closely related random regime ($D$-regular instances with random parities), Ref.~\cite{MarwahaHadfield2022} derives a closed-form expression for the approximation ratio of depth-1 QAOA, showing that it achieves $1/2 + c_k/\sqrt{D}$ with exactly computable constants $c_k$ that rise monotonically with~$k$. For random $(k,D)$-regular instances, they also establish a theoretical bound of the optimum concentrating around $1/2 + \Theta(\sqrt{k/D})$, which confirms that the $1/\sqrt{D}$ scaling reflects the structure of optimal solutions themselves and not just algorithmic limitations.

\paragraph{Summary.}
The worst-case scaling for bounded-degree \maxekxor{$k$} is understood up to $k$-dependent constants: current algorithms and hardness results together establish a $\Theta_k(1/\sqrt{D})$ scaling. 
It should be noted, however, that for large even $k$ on instances containing triangles, the algorithmic picture remains incomplete as Barak et al.'s algorithm~\cite{Barak2015} requires odd $k$ or triangle-freeness.
The same $1/\sqrt{D}$ scaling also appears in average-case analyses of random Gallager instances, though no formal connection between worst-case hardness and average-case performance is known. Any quantum advantage for worst-case bounded-degree \maxekxor{$k$} in this setting must appear through the constant prefactor, not through the asymptotic dependence on the degree bound.

\subsection{Bounded-degree \maxlin}
For $q > 2$, the landscape is qualitatively different: the hardness side is established by \cref{thm:bounded_degree_linsat} (and \cref{cor:arity,cor:bounded_degree_general}) at $r/q + \calO_{q,r}(1/\sqrt{D})$, but the algorithmic side lags behind. The only worst-case guarantee for general alphabets remains H\r{a}stad's greedy method~\cite{Hastad2000} which extends routinely to CSPs over any fixed finite abelian group (in particular $\mathbb{F}_q$), achieving $r/q + \Omega(1/D)$---a full quadratic factor weaker in $D$ than the hardness ceiling.

The algorithm of Barak et al.~\cite{Barak2015} is inherently Boolean, relying on $\{\pm 1\}$ Fourier analysis and Bonami--Beckner hypercontractivity, and does not appear to extend in a straightforward way to $\mathbb{F}_q$ with $q > 2$. Extending it would require nontrivial generalizations of Boolean hypercontractive and invariance-based arguments. No $\mu + \Omega(1/\sqrt{D})$ (with $\mu$ being the random assignment threshold) algorithm for bounded-degree CSPs over non-Boolean alphabets is currently known. Whether any algorithm---classical or quantum---can match the $1/\sqrt{D}$ hardness scaling of \cref{thm:bounded_degree_linsat} (and \cref{cor:arity,cor:bounded_degree_general}) for $q > 2$ remains open (see Discussion).

\subsection{Information-theoretic limitations of DQI in the 
bounded-degree setting}
\label{rem:dqi_ceiling}

The algorithmic gap for $q > 2$ identified in the previous subsection is particularly interesting from the perspective of DQI. Ref.~\cite[Sections~13.1--13.2]{Jordan2024DQI} establishes information-theoretic ceilings on DQI's performance for max-XORSAT under both classical and quantum decoding; we extend these to \maxlinqr{$q$}{$r$} and derive their implications for bounded-degree instances.

The DQI semicircle law~\cite{Jordan2024DQI}, combined with information-theoretic decoding limits, constrains DQI's approximation ratio on bounded-degree instances independently of any complexity assumption. These bounds are complementary to the complexity-theoretic hardness results above: they constrain the DQI framework itself, regardless of computational efficiency.
Consider a \maxeklinsat{$k$}{$q$}{$r$} instance with constraint matrix $B \in \mathbb{F}_q^{m \times n}$ of row weight $k$ and column weight at most $D$. The dual code $C^\perp = \ker(B^T) \subseteq \mathbb{F}_q^m$ has rate $R = 1 - \operatorname{rank}(B)/m \geq 1 - n/m$ (with equality generically for full-column-rank instances). 
We compute the DQI ceilings on $(k,D)$-regular instances, the natural setting in which $n/m=k/D$. Our hardness results (\cref{thm:bounded_degree_linsat,cor:arity}) establish hardness for the broader class of degree-$\leq D$ instances; since $(k,D)$-regular instances form a restriction of this class, these bounds serve as the relevant complexity-theoretic reference scale rather than a proven ceiling on the regular subclass itself. The comparisons below are therefore between the $D$-scalings.
Expanding the semicircle law~\cite{Jordan2024DQI} for small decoding radius $\ell/m$ gives
\begin{equation}\label{eq:semicircle_expansion}
    \alpha_{\mathrm{DQI}} = \left(
      \sqrt{\frac{\ell}{m}\!\left(1 - \frac{r}{q}\right)}
    + \sqrt{\frac{r}{q}\!\left(1 - \frac{\ell}{m}\right)}
  \right)^{\!2} = \frac{r}{q} 
    + \frac{2\sqrt{r(q-r)}}{q}\sqrt{\frac{\ell}{m}} 
    + \calO\!\left(\frac{\ell}{m}\right),
\end{equation}
if $r/q \leq 1 - \ell/m$, and $\alpha_{\mathrm{DQI}} = 1$ otherwise. Here, $\ell$ denotes the asymptotic decoding radius under the induced channel model, i.e.,  the maximum number of random errors in $C^\perp$ that the decoder can reliably correct.

\begin{enumerate}
    \item \emph{Classical decoders.} For the binary case ($q = 2$), Jordan et al.~\cite[Theorem~13.2]{Jordan2024DQI} show that DQI with classical decoders is bounded by
    \begin{equation}
        \alpha_{\mathrm{DQI}}^{\mathrm{cl}} \leq \frac{1}{2} 
        + \sqrt{\frac{n/m}{\log(m/n)}},
    \end{equation}
    using the Shannon capacity of the binary symmetric channel 
    \cite{Shannon1948,CoverThomas}. 
    The $\log(m/n)$ denominator arises because $C_2(p) = 1 - H_2(p) \sim 1 - p\log_2(1/p)$ for small $p$~\cite{Jordan2024DQI}, so inverting the rate condition $1 - n/m \leq C_2(\ell/m)$ yields $\ell/m = \calO((n/m)/\log(m/n))$. For general $q$, the same argument applied to the $q$-ary symmetric channel~\cite{Shannon1948,CoverThomas} gives $C_q(p) = \log_2 q - H_q(p)$ with $H_q(p) \sim p\log_2((q-1)/p)$ for small $p$, yielding
    \begin{equation}
        \alpha_{\mathrm{DQI}}^{\mathrm{cl}} \leq \frac{r}{q} 
        + \calO_{q,r}\!\left(\sqrt{\frac{n/m}{\log(m/n)}}\right),
    \end{equation}
    which on $(k,D)$-regular instances specializes to $r/q + \calO_{q,r}(1/\sqrt{D \log D})$. This would improve on H\r{a}stad's $r/q + \Omega(1/D)$~\cite{Hastad2000} but falls short by a $\sqrt{\log D}$ factor of the $r/q + \calO(1/\sqrt{D})$ that \cref{cor:arity} establishes for bounded-degree instances.

    \item \emph{Quantum decoders.} Quantum decoders may partially evade the classical $\log(1/p)$ entropy penalty, because coherent decoding accesses non-orthogonal quantum states rather than classical syndrome bits. For the binary case ($q = 2$), Jordan et al.~\cite[Section~13.2]{Jordan2024DQI} model the coherent errors in DQI as a classical-quantum channel and show via the Holevo bound~\cite{holevo1973} that
    \begin{equation}
        \alpha_{\mathrm{DQI}}^{\mathrm{q}} \leq 1 
        - H_2^{-1}\!\left(1 - \frac{n}{m}\right).
    \end{equation}
    For $(k,D)$-regular instances with $n/m = k/D$, the expansion $H_2(1/2 - \delta) = 1 - \Theta(\delta^2)$ near capacity gives $1/2 + \Theta(1/\sqrt{D})$, matching the $1/\sqrt{D}$ scaling of the bounded-degree hardness bound of \cref{cor:arity}.
    The key reason the $\sqrt{\log D}$ gap disappears is that the Holevo capacity of the induced channel satisfies $\chi(p) = 1 - \Theta(p)$ for small $p$~\cite{Jordan2024DQI}, losing capacity linearly in $p$ rather than as $p \log(1/p)$.
    For general $q$, 
    the natural $q$-ary generalization maps each symbol $i \in \mathbb{F}_q$ to the state vector $|i_p\rangle = \sqrt{1-p}\,|i\rangle + \sqrt{p/(q-1)}\sum_{j \neq i}|j\rangle$ with $p = \ell/m$. 
    Since the outputs are pure, the Holevo capacity equals the von Neumann entropy of the average state $\bar{\rho} = \frac{1}{q}\sum_i |i_p\rangle\langle i_p|$. A direct computation gives $\chi_q(p) = \log_2 q - 2p/\ln 2 + \calO(p^{3/2})$, with the leading coefficient independent of $q$---again linear in $p$, without the logarithmic penalty.
    The rate condition then yields $\ell/m = \Theta(n/m)$, and via \eqref{eq:semicircle_expansion}:
    \begin{equation}
        \alpha_{\mathrm{DQI}}^{\mathrm{q}} \leq \frac{r}{q} 
        + \calO_{q,r}\!\left(\frac{1}{\sqrt{D}}\right)
    \end{equation}
    on $(k,D)$-regular instances, matching the $1/\sqrt{D}$ scaling of \cref{cor:arity} for all $(q,r)$.
\end{enumerate}

Thus, \emph{DQI with classical decoders} faces an information-theoretic $\sqrt{\log D}$ limitation on bounded-degree instances. By contrast, \emph{DQI with quantum decoders} is information-theoretically compatible with the $1/\sqrt{D}$ scaling: for $q = 2$ by following Ref.~\cite{Jordan2024DQI}, and via the $q$-ary channel generalization above for $q > 2$. 

This underscores the importance of quantum decoding as a direction for DQI: As observed in Ref.~\cite{shutty2026optimizationusinglocallyquantumdecoders}, quantum decoders significantly broaden the algorithmic toolbox for bounded degree instances, and our results show that this broadening is also justified from a complexity-theoretic and information-theoretic viewpoint. To the best of our knowledge, no efficient quantum decoder achieving the Holevo-capacity scaling for the codes arising in bounded-degree DQI instances is currently known, though several approaches toward \emph{efficient quantum decoding} have been developed~\cite{Renes2017,BMP2022,buzet2025finegrainedunambiguousmeasurements,chailloux2024softdecoders,mandal2026beliefpropagationquantummessages,shutty2026optimizationusinglocallyquantumdecoders,piveteau2025_quantum_decoding}.

\section{Discussion}
\label{sec:discussion}

In this work, we have explored and identified the best approximation ratios achievable by polynomial-time algorithms for an important family of optimization problems.
The central message of this work is that the $1/\sqrt{D}$ scaling observed across all known algorithms for bounded-degree \maxkxor{} is complexity-theoretically forced, and that the same scaling extends to \maxlinqr{$q$}{$r$} over arbitrary finite fields.
We conclude by highlighting several concrete questions that remain open in light of the bounded-degree inapproximability results established above.

\begin{enumerate}
    \item \textbf{Optimal worst-case constant.} What is the tight worst-case constant $c^*(k)$ such that $1/2 + c^*(k)/\sqrt{D}$ is achievable but $1/2 + (c^*(k) + \varepsilon)/\sqrt{D}$ for any $\varepsilon > 0$ is $\NP$-hard? The algorithm of Barak et al.~\cite{Barak2015} gives $c^*(k) \geq \exp(-\calO(k))$, but  the true constant could be much larger. Equivalently, is the $\exp(-\calO(k))$ dependence necessary, or can it be improved to a polynomial in $k$?

        \item \textbf{Closing the gap for bounded-degree \maxlinqr{$q$}{$r$} with $q > 2$.} \Cref{cor:arity} establishes hardness at $r/q + \calO_{q,r}(1/\sqrt{D})$, but the best known worst-case algorithm for $q > 2$ remains H\r{a}stad's greedy method~\cite{Hastad2000} at $r/q + \Omega(1/D)$---a quadratic gap in $D$. The algorithm of Barak et al.~\cite{Barak2015} is inherently Boolean, relying on $\{\pm 1\}$ Fourier analysis and Bonami--Beckner hypercontractivity. A classical route to closing the gap would be to extend their technique to $\mathbb{F}_q$.
    A more speculative, and potentially quantum, route is suggested by the information-theoretic arguments made in \cref{rem:dqi_ceiling} that suggest that DQI combined with a quantum decoding algorithm might feature the $r/q + \calO(1/\sqrt{D})$ scaling for bounded-degree instances. However, no quantum decoding algorithms are currently known that match the information-theoretically achievable decoding limit, which would yield this performance.
    We note that worst-case bounded-degree instances can have dual codes with minimum distance $\calO(1)$, in which case DQI reduces to random assignment regardless of decoder quality. The information-theoretic analysis of \cref{rem:dqi_ceiling} shows that even on instances with good coding-theoretic properties, DQI with classical decoders is limited to $r/q + \calO(1/\sqrt{D \log D})$. 
    We further note that our hardness is established for degree-at-most-$D$ instances; sharpening it to exactly $(k,D)$-regular instances---the setting of the average-case ensembles---remains open.

    \item \textbf{Quantum constant advantage for $k \geq 3$.} This raises the important question relevant for assessing potential quantum advantages \cite{MindTheGaps,VastWorld,Abbas_2024}: Can any quantum algorithm achieve a strictly larger constant $c$ than the best classical algorithm on bounded-degree \maxekxor{$k$}? On average-case bounded-degree $(k,D)$-regular \maxekxor{$k$}-instances drawn from the Gallager-ensemble, Ref.~\cite{shutty2026optimizationusinglocallyquantumdecoders} shows that Regev+FGUM is exactly matched by the classical Turbo Prange for all $(k, D)$. QAOA achieves larger constants at moderate depth, but it is unknown whether this advantage persists against all classical algorithms. For the related but distinct problem of MaxCut ($k = 2$) on $D$-regular large-girth graphs, the constant-prefactor competition is more advanced: the best known assumption-free classical algorithms achieve $c = 2/\pi \approx 0.637$~\cite{BarakMarwaha2022}, while QAOA at depth $p = 11$ provably exceeds this~\cite{Basso2021}. A conditional classical algorithm also exceeds $2/\pi$, but only assuming the \emph{no overlap gap property}~\cite{ElAlaoui2021}. However, MaxCut is qualitatively different---the cost function is quadratic, and the classical benchmark arises from SDP rounding rather than the hypercontractivity-based methods relevant to $k \geq 3$---so these results do not transfer directly. We stress that this competition over constants is not a second-order concern: since the approximation ratio lives in $[r/q, 1]$, the constant prefactor $c$ is precisely what an algorithm buys over random assignment, so a quantum algorithm with a strictly larger constant $c$ would constitute a genuine, operationally meaningful advantage on these instances.

    \item \textbf{Average-case vs.\ worst-case.} The average-case constants on Gallager instances vastly exceed the worst-case constant $\exp(-\calO(k))$. Does a complexity-theoretic explanation exist for why Gallager instances are easier, or is the gap an artifact of the analysis?
\end{enumerate}
We hope that this work can stimulate the answer of some of those questions.

\section*{Acknowledgements}
The authors warmly thank Franz Schreiber, Asad Raza, Elies Gil-Fuster, Alexander Nietner, and Jean-Pierre Seifert for stimulating discussions and feedback on this work. Moreover, the authors acknowledge support by the BMFTR (PraktiQOM, QuSol, HYBRID++, PQ-CCA), the Munich Quantum Valley, Berlin Quantum, the QuantERA, the Quantum Flagship (Millenion, PasQuans2), the DFG (CRC 183, SPP 2514), the Clusters of Excellence (MATH+, ML4Q), and the European Research Council (DebuQC).

\bibliographystyle{quantum}
\bibliography{main}
\end{document}